\documentclass[11pt]{article}

\RequirePackage[OT1]{fontenc}

\RequirePackage[aap,amsthm,amsmath,natbib]{imsart}
 
\usepackage{times}
\usepackage{bm}
\usepackage{graphicx}
\usepackage{amssymb}
\usepackage{amsmath}
\usepackage{color}
\def\path{/usr/local/teTeX/share/texmf.local/tex/latex/base}

\startlocaldefs
\usepackage{url}
\newtheorem{theorem}{Theorem}[section]
\newtheorem{lemma}[theorem]{Lemma}
\newtheorem{algorithm}{Algorithm}[section]

\newtheorem{remark}{Remark}

\def\rr{{\mathbb R}}

\def\Gap{{\mathbf{Gap}}}
\def\h{{\mathbf h}}
\def\min{\mbox{min}}

\def\ii{{(i)}}

\def\xx{\mathbb{R}^d}
\def\E{\mathbb{E}}
\def\Unif{\emph{Unif}}

\endlocaldefs

\begin{document}

\begin{frontmatter}

\title{Small World MCMC with Tempering: \protect\\ Ergodicity and Spectral Gap}
\runtitle{STEEP sampler: Ergodicity and Spectral Gap }

%
%
%
%

\begin{aug}
\author{\fnms{Yongtao} \snm{Guan}\corref{}\thanksref{t1}\ead[label=e1]{yongtaog@bcm.edu}}
\and
\author{\fnms{Matthew} \snm{Stephens}\ead[label=e2]{mstephens@uchicago.edu}}
\runauthor{Guan and Stephens}

\address{Baylor College of Medicine \\Children's Nutrition Research Center   \\1100 Bates Room 2070 \\ Houston, TX 77030}
\address{University of Chicago \\Eckhart Hall Room 126\\5734 S. University Avenue\\Chicago, IL 60637}
\address{\printead{e1}}
\address{\printead{e2}}
\affiliation{Baylor College of Medicine and University of Chicago}

\thankstext{t1}{Thanks to Winfried Barta for helpful discussions regarding ergodicity proof.}

\end{aug}

\begin{abstract}
When sampling a multi-modal distribution $\pi(x)$,  $x\in \rr^d$, a Markov chain with local proposals is often slowly mixing; while a Small-World sampler \citep{guankrone} --  a Markov chain that uses a mixture of local and long-range proposals -- is fast mixing.  However, a Small-World sampler suffers from the curse of dimensionality because its spectral gap depends on the volume of each mode. We present a new sampler that combines tempering, Small-World sampling, and producing long-range proposals from samples in companion chains (e.g. Equi-Energy sampler).  In its simplest form the sampler employs two Small-World chains:  an exploring chain and a sampling chain.  The exploring chain samples $\pi_t(x) \propto \pi(x)^{1/t}$, $t\in [1,\infty)$, and builds up an empirical distribution. Using this empirical distribution as its long-range proposal, the sampling chain is designed to have a stationary distribution $\pi(x)$. We prove ergodicity of the algorithm and study its convergence rate. We show that the spectral gap of the exploring chain is enlarged by a factor of $t^{d}$ and that of the sampling chain is shrunk by a factor of $t^{-d}$. Importantly, the spectral gap of the exploring chain depends on the ``size" of $\pi_t(x)$ while that of sampling chain does not.  Overall, the sampler enlarges a severe bottleneck at the cost of shrinking a mild one, hence achieves faster mixing. The penalty on the spectral gap of the sampling chain can be significantly alleviated when extending the algorithm to multiple chains whose temperatures $\{t_k\}$ follow a geometric progression. If we allow $t_k \rightarrow 0$, the sampler becomes a global optimizer.
 
\end{abstract}



\end{frontmatter}

\section{Introduction}  
Developing an algorithm to improve sampling efficiency for a high dimensional multi-modal distribution has been an active research area \citep{Geyer91,tempering,neal,ee,madras.zheng,guankrone,ee-proof,Woodard07,brockwell,delmoral}.
In this paper we introduce a new sampling algorithm and study its ergodicity and convergence properties. The new algorithm
combines several existing ideas: tempering \citep{Geyer91,tempering}, propagating information between chains via empirical distributions \citep{ee}, and Small-World sampling \citep{guan,guankrone}.  
A Small-World sampling combines local and long-range proposals in a Metropolis-Hastings algorithm. \citet{guankrone} showed that a Small-World chain converges faster than a Metropolis-Hastings chain with local proposals when sampling a multi-modal density.

The new algorithm is simple and easy to implement. 
In its simplest form the algorithm has two Small-World chains: the first chain samples a fattened (or flattened) distribution (via tempering) with its long-range proposals generated by a typical heavy-tailed distribution; the second chain samples the target distribution with its long-range proposals randomly drawn from samples of the first chain. Intuitively, the first chain identifies and remembers (via its empirical distribution) modes of the distribution -- it does  so more effectively than an unheated chain because it accepts more long-range proposals that tend to move between different modes.  This knowledge of whereabouts of different modes is used by the second chain through its long-range proposals. 

The new algorithm bears similarity with the Equi-Energy sampler. Indeed, our original goal was to prove the ergodicity of the Equi-Energy sampler -- the ergodicity proof in the Equi-Energy paper was incomplete and the amended proof \citet{atchade+liu} has difficulties.  Identified these problems, \citet{ee-proof} proved ergodicity of the Equi-Energy sampler using the Poisson's Equation. By studying properties of perturbed kernels, \citet{fort} proved ergodicity of a class of adaptive MCMC algorithms, including the Equi-Energy sampler.  Our proof  uses mixingale theory,  built on the work of \cite{atchade+liu} and \cite{atchade+rosenthal}, which allows us to study the asymptotic convergence rate. The convergence results identified that local proposal in the feeding chain of Equi-Energy sampler slows down the convergence in a multi-modal distribution. Thus we propose using a small-world sampler as the feeding chain.  

The paper contributes both practically and theoretically in sampling high dimensional multi-modal distributions.  First, we provide a simple and easy-to-implement algorithm that is effective for challenging problems.   
Second, we prove ergodicity of the algorithm and analyze its convergence rate. The result on convergence rate provides important insights into when, why and how the algorithm will improve convergence in practice.  

\medskip

We first formally set up the problem. Let $(\xx, \mathcal{B}(\xx))$ be the state space equipped with its $\sigma$-algebra. 
 A Metropolis-Hastings algorithm \citep{metropolis,Hast70} aims to sample from a probability measure $\pi$ that admits a density with respect to Lebesgue measure that is only known up to a normalizing constant. We use $\pi$ to denote a measure and $\pi(x)$ to denote a density and it should be clear from the context.  
The transition kernel of a Metropolis-Hastings chain is
\begin{equation}\label{eqn:P}
P(x,dy)=a(x,y)\:k(x,dy) +r(x)\:\delta_x(dy),
\end{equation}
where $k(x,y): \xx \times \xx \rightarrow [0, \infty)$ is a \emph{proposal distribution}, 
$a(x,y)=\min \left(1, \frac{\pi(y)}{\pi(x)} \frac{k(y,x)}{k(x,y)}\right)$
is the \textit{acceptance probability} of a proposed move $y$, 
the $\delta_x$ is the point mass at $x$, and
$r(x)=\int_{\xx}{(1-a(x,y))k(x,dy)}$
is the probability that $y$ being rejected. 
Obviously $k(x,y)$ determines $P(x,dy)$.

Let $f : \xx \rightarrow \rr$ denotes a measurable function. 
For a signed measure $\mu$ and a positive function $V$, define the $V$-norm of $\mu$ by
$||\mu||_V := \sup_{f \le V}{|\mu (f)|} $
where $\mu(f) = \int_{\xx}{f(x) \mu(dx)}$.
A transition kernel $P$ acts on $f$ such that 
$Pf(x)=\int_{\xx}{f(y)P(x,dy)}$.
Given two transition kernels $P, Q$, the product $PQ$ is also a transition kernel $(PQ)(x,A) = \int_{\xx}{P(x, dy) Q(y, A)}$. This allows us to define a product of kernels of countable many through induction.

\subsection{Minorization and drift conditions}
We assume the following assumptions hold:
\begin{enumerate}
\item[A1] A probability measure $\psi$ exists on $\mathcal{B}$ such that the Markov chain is $\psi$-irreducible and aperiodic \cite[c.f.][]{meynbook}.
\item[A2] \emph{Minorization Condition}: there exist a set $C \in \mathcal{B}$ and $\epsilon > 0$ such that, for the same $\psi$ in A1, we have $\psi(C) > 0$ and, for all $A \in \mathcal{B}$, $x \in C$,
\begin{equation} \label{a:min}
P(x, A) \ge \epsilon \;\psi(A).
\end{equation}
\item[A3] \emph{Drift Condition}: there exist a measurable function $V: \xx \rightarrow [1, \infty)$ such that $\pi(V) < \infty$ and constants $\lambda \in [0, 1)$ and $b \in (0, \infty)$ and same set $C$ in A2 satisfying
\begin{equation} \label{a:drift}
P V(x) \le \lambda V(x) + b {\bf1}_C(x).
\end{equation}
\end{enumerate}
The minorization condition ensures the \emph{flow} from the set $C$ to outside is lower-bounded, an idea intimately connected to conductance (see Section \ref{sec:gap}).   The drift condition guarantees that the Markov chain evolves towards $C$.  
The drift condition is necessary to define $V$-geometrically ergodic.
The above minorization and drift conditions can be checked for many practical problems. For example, if $P$ is a Random Walk
Metropolis kernel, both \eqref{a:min} and \eqref{a:drift}  are known to hold under some regularity conditions on the target densities \cite[see][]{atchade10}. 

\subsection{Spectral gap}
A homogeous Markov chain that satisfies (A1-A3) is geometric ergodic \citep[c.f.][]{meynbook,rrsurvey}.  Let $L^2(\pi)$ denote the space of measurable functions on $\xx$ with $\int_{\xx}{f(x)^2\pi(dx)}<\infty,$
 with inner product $\langle f,g\rangle =\int_{\xx}f(x)g(x)\pi(dx)$, and norm $\|f\|=\langle f,f\rangle^{1/2}$.
The operator $P$ being reversible with respect to $\pi$ is equivalent to $P$ being self-adjoint.
It is well known that the spectrum of $P$ is a subset
of $[-1,1].$ (Self-adjoint implies its spectrum is real,
and being a Markov transition kernel determines the
range.)
A chain is said to be $L^2(\pi)$-geometrically ergodic if there exist
a constant $\rho \in (0,1)$ and a positive $M < \infty$, and $V(x)$ defined in (A3) such that
\begin{equation} \label{eqn:geometric}
\|P^n(x, \cdot)-\pi(\cdot)\|_V\le M \rho^n V(x).  
\end{equation}
Define $L_0^2(\pi)=\{f\in L^2(\pi): \langle f,{\bf1} \rangle
=0\}.$
Denote by $P_0$ the restriction of $P$ to $L_0^2(\pi)$. 
It has been shown \citep{rr,roberts} that for
reversible Markov chains, geometric ergodicity is equivalent to
the condition
\begin{equation}\label{eqn:P0}
|||P_0||| \equiv \sup_{f\in L_0^2(\pi), \|f\|\le 1}\|P_0 f \| < 1.
\end{equation}
The \textit{spectral gap} of the chain $P$ is
defined by
\begin{equation}
\Gap(P)=1-|||P_0|||.
\end{equation}
The spectral gap determines the convergence speed of a MCMC algorithm. Very roughly, a chain is close to equilibrium after a few multiples of $1/\Gap(P)$ iterations \citep{madras}.

\subsection{Piecewise log-concave distributions} \label{sec:logcon}
We assume the target density $\pi(x)$ is a mixture of log-concave densities.
A function $f: \xx \rightarrow (0, \infty)$ is \textit{log-concave} if for any $s \in [0, 1]$,
\begin{equation}
f(s\: x + (1-s) y) \ge f(x)^s f(y)^{1-s}.
\end{equation}
Let $|\cdot|$ be a metric on $\xx$, $f$ is \textit{$\alpha$-smooth} if  $|\log{f(x)}-\log{f(y)}| < \alpha \:|x-y|$ for all $x, y \in \xx$.  
Let $\{A_1, \dots, A_m\}$ be a partition of $\xx$.  
Let $\pi_{(i)}:A_i \rightarrow (0, \infty)$ be an $\alpha$-smooth log-concave density with \emph{barycenter} $\beta_i =\int_{A_i}{\pi_{(i)}(dx)}$ and the \emph{first moment} $M_i = \int_{A_i} {|x-\beta_i| \pi_{(i)}(dx)}$.
Let $\beta_{ij}=|\beta_i- \beta_j|, \; i\neq j$ and $\beta_\pi = \max(\beta_{ij})$.
Let $w_i > 0$, the distribution of interest is 
\begin{equation}\label{eqn:pi}
\pi(x) \propto \sum_{i=1}^m \pi_{(i)}(x) 1_{A_i}(x) \;w_i.
\end{equation}
Define average proposal distance of $k(x, y)$ as $D= \int_{\xx\times\xx}{|x-y| k(x,y) dx dy}$.  Then $k(x,y)$ is \emph{local} if $D < \min\{M_i\}$ and \emph{long-range} if $D >  \beta_{\pi}$.  We call a chain \emph{local chain} if it only uses local proposals.
By  \emph{size} (or \emph{complexity}) of $\pi(x)$ we mean the quantities associated with $\pi(x)$ that will affect the spectral gap. Namely, the measure of the steepness of each mode $\alpha$ and the measure of the distances between modes $\beta_{ij}$.  We treat dimensionality of $\xx$ and number of modes $m$ as fixed quantities.

\subsection{State decomposition theorem}\label{sec:SDT}
We describe the ``pieces" of a Metropolis--Hastings
chain $P$ by defining, for each $i=1,\dots, m$, a new Markov chain
on $A_i$ that rejects any transitions of $P$ out of $A_i$. The
transition kernel $P_{A_i}$ of the new chain is given by
\begin{equation}
P_{A_i}(x, B)= P(x,B)+1_B(x)P(x, A_i^c) \hspace{.2in} \mbox{for }
x\in A_i, B\subset A_i. \label{KernelPAi}
\end{equation}
The movement of the original chain among the ``pieces" can be
modeled by a ``component" Markov chain with state space $\{1,
\ldots, m\}$ and transition probabilities:
\begin{equation} \label{eqn:ph}
P_c(i,j)=\frac{1}{2\; \pi(A_i)}\int_{A_i}P(x, A_j)\pi(dx),
\hspace{.2in}\text{for } i\neq j,
\end{equation}
and $P_c(i, i)= 1-\sum_{j\neq i}P_c(i,j)$.  

\begin{theorem} [State Decomposition Theorem]\label{thm:SDT}
In the preceding framework, as given by Equations
(\ref{KernelPAi}) and (\ref{eqn:ph}), we have
\begin{equation}
\Gap (P)\geq \frac{1}{2}\Gap (P_c)({\bf min}_{i=1,\dots,m}\Gap
(P_{A_i})). \label{SDT}
\end{equation}
\end{theorem}
\citet{guankrone} proved this state decomposition theorem, generalized it from its original version \citep{madras}.   
The theorem says that we can bound the spectral gap by taking into
account of the mixing speed within each mode and that of
among different modes.
In Section \ref{sec:gap}, we will focus on $\Gap(P_c)$ because it causes slowly mixing. 
 
\subsection{Tempering and Equi-Energy sampler}
Since \citet{Geyer91} and \citet{tempering}, \emph{tempering} has become a popular technique. Here we use two chains to illustrate tempering, although it usually employs multiple chains \citep{Geyer91}.
\begin{algorithm}
Let $s \in (0,1)$. Let $x_c$ and $x_h$ be current states of the cold and hot chains that sample $\pi(x)$ and $\pi_t(x) \propto \pi(x)^{1/t}$ (for some $t>1$) respectively. Repeat the following steps.
\begin{enumerate}
\item Simulate $u\sim \Unif\:[0,1)$.
\item If $u < s$, update $x_c$ and $x_h$ independently using the Metropolis-Hastings algorithm.
\item If $u \ge s$, compute
$a=\frac{\pi_t(x_h)}{\pi_t(x_c)} \frac{\pi(x_c)}{\pi(x_h)}$
and swap $x_c$ and $x_h$ with probability $\min(1,a).$
\end{enumerate}
\end{algorithm}
The hot chain samples a flattened distribution; the flattening makes it easier for a local chain to move around to discover new modes. This easiness in the hot chain transfers to the cold chain through successful ``swap".  However, tempering has at least two limitations. First, because the process is ``memory-less", modes need to be repeatedly rediscovered. Second, the cold chain interferes with the hot chain because of the ``swap", which may slow down mixing. The Equi-Energy sampler \citep{ee} resolves both limitations  using empirical distributions.     
The Equi-Energy sampler runs multiple chains of different temperatures and samples of each chain are recorded and classified into different energy rings according to their energy level (log density). In addition to local proposals, the Equi-Energy sampler has ``equi-energy jumps", when a lower temperature chain proposes a new move by randomly drawing a sample from an energy ring in a higher temperature chain. Thus, the relative easiness of moving between different modes in hot chains propagates down to the cold chain. Note a high-temperature chain affects the proposal of a lower-temperature chain; while a lower-temperature chain does not affect a high-temperature chain.

\medskip
The rest of the paper is organized as following.  In the next section, we describe the algorithms and main results and compare the algorithm with the Equi-Energy sampler.  
In Section \ref{sec:ergodic} we prove ergodic theorems  and discuss the convergence rate.  In Section \ref{sec:gap} we prove theorems regarding spectral gaps of the algorithm.  In Section \ref{sec:app} we discuss practical issues and applications.  A short discussion will conclude the paper.

\section{Algorithms and Main Results}
The new algorithm combines a Small-World sampler with tempering and we call it ``Small-world Tempering with Empirical Ensemble Propagation", or STEEP. It has a backronym ``Small-world Tempering with Equi-Energy Program" to credit the Equi-Energy sampler.   
We begin with the Small-World sampler. 
\begin{algorithm}[Small World Sampler] \label{alg:sw}
Let $s \in (0,1)$ and the current state is $X_n=x$.
Let $l(x,y)$ and $h(x,y)$ be local and long-range proposals respectively.
\begin{enumerate}
\item With probability $(1-s)$ simulate $y$ from density $l(x,y)$ and compute $a =\frac{\pi(y)}{\pi(x)}\frac{l(y, x)}{l(x,y)}$.
\item Otherwise, simulate $y$ from density $h(x,y)$ and compute $a =\frac{\pi(y)}{\pi(x)}\frac{h(y, x)}{h(x,y)}$; \label{long-range-step}
\item Set $X_{n+1} \leftarrow y$ with probability $\min(1,a)$, otherwise set $X_{n+1} \leftarrow x$. 
\item Set $n\leftarrow n+1$, goto step 1.  
\end{enumerate}
\end{algorithm}
When sampling a multi-modal density, both simulation studies \citep{guan} and theoretical work \citep{guankrone} have shown that a Small-World sampler converges faster than a local chain.  
Intuitively, the local proposals allows the chain to better explore a mode, while the long-range proposals allows it to jump between modes more easily.
However, if modes are steep and far apart, a Small-World sampler will fail because that successful transitions between modes (initiated by long-range proposals) depend on their volumes \citep{guankrone} -- to increase volumes of modes, tempering does the trick. 

\subsection{The STEEP Algorithm}
 The simplest form of STEEP employs two chains that run simultaneously.  Both are Small-World chains because their proposal distributions are mixtures of local and long-range. The \emph{exploring chain}  samples a fattened target $\pi_t(x) \propto \pi(x)^{1/t}$ for some $t>1$, where the long-range proposals are drawn from a typical heavy-tailed distribution (e.~g., Cauchy).  Samples collected in the exploring chain induce an empirical distribution $\xi$.  The \emph{sampling chain} samples $\pi(x)$, using $\xi$ to simulate its long-range proposals. Because the similarity between $\pi_t$ and $\pi$, in particular,  their modes are in the same places, the empirical distribution provides natural and ``intelligent"  long-range proposals for the sampling chain. 
In detail:
\begin{algorithm} [STEEP with Two Chains]\label{alg:two}
Let $\pi_t (x) \propto \pi(x) ^{1/t}$ for some $t>1$. $\{Y_n\}$ and $\{X_n\}$ sample $\pi_t$ and $\pi$ respectively, and $Y_n=y_n$ and $X_n = x_n$.  
\begin{enumerate}
\item Simulate $Y_{n+1}|Y_n=y_n$ using Algorithm \ref{alg:sw} to obtain $y_{n+1}$.
Update empirical measure $\xi_{n+1} = n/(n+1) \xi_n + 1/(n+1)\delta{(y_{n+1})}$. 
\item Simulate $X_{n+1}|X_n=x_n$ using Algorithm \ref{alg:sw}.  
Modify step (\ref{long-range-step}) so that $y \sim \xi_{n+1}$ and  compute
$a =\frac{\pi(y)}{\pi(x_n)}\frac{\pi_t(x_n)}{\pi_t(y)}.$
Set $X_{n+1} \leftarrow y$ with probability $\min(1,a)$, otherwise set $X_{n+1} \leftarrow x_n$. 
\item Set $n\leftarrow n+1$, goto step 1. 
\end{enumerate}
\end{algorithm}
The algorithm \ref{alg:two} is similar to the algorithm in section $3.3$ of \cite{andrieu+etal}. The difference here is that we emphasize the long-range proposals. 

\begin{remark}[Ergodicity]
The exploring chain is a homogenous Markov chain with stationary distribution $\pi_t$. Conditional on the exploring chain, the sampling chain is a non-homogeneous Markov chain, whose transition kernel evolves over time because it depends on $\xi_n$.
 Since $\xi_n \rightarrow\pi_t$, the sampling chain converges to a Small-World chain with long-range proposal  $\pi_t$ and stationary distribution $\pi$. As a result, if we run Algorithm \ref{alg:two} long enough the sampling chain should generate samples approximately from $\pi$. This intuition is formalized in the ergodic theorem in Section \ref{sec:ergodic}.
\end{remark} 

Since the exploring chain samples a fattened distribution $\pi_t$, each mode is larger and hence more easily found by long-range proposals. This implies a better mixing compared to directly sampling $\pi$. We quantify the intuition in the following theorem.

\begin{theorem} \label{thm:explore}
Denote $E$ the exploring chain in Algorithm \ref{alg:two} with stationary distribution $\pi_t$, and denote $E_c$ the component chain of $E$ (see Section \ref{sec:SDT}), then $\Gap(E_c) \ge c_1\; t^d$ for some constant $c_1$ that is independent of $t$ and $d$.
\end{theorem}

However, a large $t$ increases dissimilarity between $\pi_t$ and $\pi$.  Large dissimilarity reduces acceptance ratio of long-range proposals in the sampling chain -- causing slow mixing. We quantify this intuition in the following theorem.  

\begin{theorem} \label{thm:sample}
Denote $S$ an (idealized) sampling chain in Algorithm \ref{alg:two}, using $\pi_t(x)$ as its long-range proposal instead of $\xi$,  and denote $S_c$ the component chain of $S$ (see Section \ref{sec:SDT}),  then $c_2 t^{-2d} \le \Gap(S_c) \le c_3\; t^{-d}$ for some constants $c_2, c_3$. In addition, $\Gap(S_c)$ is independent of the size of $\pi(x)$.
\end{theorem}

\begin{remark}[Equilibrium Assumption]
Note that Theorem \ref{thm:sample} considers a Small-World chain that uses $\pi_t$ as its long-range proposal. The connection with Algorithm \ref{alg:two} is that, as the number of iterations increases, the long-range proposal $\xi_n$ used by the Algorithm becomes an increasingly close approximation to $\pi_t$
(indeed $\xi_n \rightarrow \pi_t$ weakly). Thus, intuitively, the result in the theorem represents the ``asymptotic" behavior of the algorithm. 
\end{remark}
 
Obviously, for Algorithm \ref{alg:two} to converge quickly, both the exploring chain $E$ and the sampling chain $S$ must converge quickly. Therefore,
we want to increase $g=\min(\Gap(E_c), \Gap(S_c))$. 
Suppose we set $t=1$ in Algorithm \ref{alg:two},  then the (idealized) sampling chain converges instantaneously (because it proposes from the target distribution). We have $g=\Gap(E_c)$, which is spectral gap of a Small-World chain. 
When we increase $t$ in Algorithm \ref{alg:two}, we enlarge a small spectral gap, $\Gap(E_c)$, and pay a penalty by shrinking a large spectral gap, $\Gap(S_c)$, to achieve faster mixing.  This trade-off works well because $\Gap(E_c)$ depends on target density and can be extremely small; while $\Gap(S_c)$ is independent of the target density.  However, when $d$ is large, even for a modest $t$, the penalty on $\Gap(S_c)$ can be large. This lead us to extend the algorithm to multiple chains, instead of just two. 
 
\begin{algorithm}[STEEP] \label{alg:many}
Let $\{t_i\}$ be a sequence of positive numbers with $1= t_0  < t_1 < \cdots < t_H$.  Let $\{X^\ii_n\}$ sample $\pi_{i}(x) \propto \pi(x)^{1/{t_i}}$.  Let $\xi^\ii$ be the empirical measures by the $i$-th chain.  We call the $H$-th chain \emph{exploring} chain, the rest \emph{sampling} chains.
\begin{enumerate}
\item Simulate $X^{(H)}_{n+1} | X^H_{n} = x^{(H)}_i$ with Algorithm \ref{alg:sw} to obtain $x^{(H)}_{n+1}$. 
Update empirical measure $\xi^{(H)}_{n+1} = n/(n+1) \xi^{(H)}_n + 1/(n+1)\delta{(X^{(H)}_{n+1})}$. 
\item For each $i = H-1, H-2, \dots, 1, 0$:
\begin{itemize}

\item Simulate $X^\ii_{n+1}|X^\ii_n=x^\ii_n$ using Algorithm \ref{alg:sw} to obtain $x^\ii_{n+1}$. 
Modify step (\ref{long-range-step}) so that $y \sim \xi^{(i+1)}_{n+1}$ and  compute
$a =\frac{\pi_i(y)}{\pi_i(x^\ii_n)}\frac{\pi_{i+1}(x^\ii_n)}{\pi_{i+1}(y)}.$
\item Set $X^\ii_{n+1} \leftarrow y$ with probability $\min(1,a)$, otherwise $X^\ii_{n+1} \leftarrow x^\ii_n$. 
\item Update empirical measure $\xi^\ii_{n+1} = n/(n+1) \xi^\ii_n + 1/(n+1)\delta{(x^\ii_{n+1})}$. 
\end{itemize}
\item Set $n\leftarrow n+1$, goto step 1. 

\end{enumerate}
\end{algorithm}

We may determine \emph{optimal} temperatures $\{t_i\}$ in Algorithm \ref{alg:many}.  
Assuming $t_H$ fixed and all chains reach equilibrium, denote $G_i$ a Markov chain that samples $\pi_{i}(x)$ with long-range proposal $\pi_{i+1}(x)$, and let $g_i$ be the spectral gap of the component chain of $G_i$, for $i=0,\cdots, H-1$.
Theorems \ref{thm:sample} implies that
\begin{equation} \label{eqn:g1}
g_0= c_0 ({t_1}/{t_0})^{- d}, g_1 = c_1 ({t_2}/{t_1})^{-d}, \dots, g_{H-1} = c_{H-1} ({t_H}/{t_{H-1}})^{- d},
\end{equation}
for some constants $c_0, \dots, c_{H-1}$.  
We want to choose $\{t_i\}$ to optimize
\begin{equation} \label{eqn:g}
g= \min(g_1, \dots, g_{H-1}).
\end{equation}

Assume $c_1 = \cdots = c_H$ in Equations (\ref{eqn:g1}), then for a fixed $t_H$, we have $g_0 = \cdots = g_{H-1}$ maximizes $g$ defined in (\ref{eqn:g}).  This implies that $\{t_i\}$ is a geometric progression because of (\ref{eqn:g1}) and assumption on $c_i$'s. 
Such a geometric progression on adjacent temperatures has been suggested in both parallel tempering \citep{physics} and the Equi-Energy sampler \citep{ee}. The former is based on the optimality of the acceptance ratio of swaps between adjacent chains and the later is based on empirical evidence.
Our argument here provides an additional justification based on spectral gap.  

The exploring chain in Algorithm \ref{alg:many} can benefit from a large $t_H$ with a modest $\tau=t_{k+1}/t_k$ and a modest $H$ thanks to the geometric progression. A large $t_H$ makes the ``exploring" easier, while the penalties spread out across sampling chains. We will discuss how to choose $t_H$ and $\tau$ in practice in Section \ref{sec:app}. 

\subsection{Finding Global Optimum}
We may extend the STEEP to find the global optimum of a density $\pi(x)$ -- tempering does the trick. Specifically, in Algorithm \ref{alg:many}, instead of stop at $t=1$, we continue the process at $t=\tau^{-k}$ for $k = 1, \cdots, C$.  As $k$ increases, each mode becomes steeper and the probability concentrates around the global maximum. The samples collected in the last chain will be very close to the global optimum.
We have the following Lemma.
\begin{lemma}
Let $t_k = \tau^k$ for $k = H, H-1, \cdots, 1, 0, -1, ... -C$, then for a sufficiently large $C$, the samples collected by the last chain in Algorithm \ref{alg:many} will be arbitrarily close to the global maximum.  
\end{lemma}
\begin{proof}
Let $\pi(x)$ be an $\alpha$-smooth log-concave density that attains its unique maximum at $0$. Assume $\pi_t(x) \propto \pi(x)^{1/t}$.  It is sufficient to show that for any $\epsilon > 0$  there exists $\delta>0$, such that for any $0<t<\delta$ we have $\int_{B_\epsilon} \pi_t(x) > 1-\epsilon$, where $B_{\epsilon}$ is a ball centered at $0$ with radii $\epsilon$.  However, the claim hold trivially because $\pi_t(x)$ decays exponentially at rate $\alpha / \delta$ which can be made arbitrarily large by letting $\delta \rightarrow 0$ (or equivalently increasing $C$).  Extension to mixture of $\alpha$-smooth log-concave densities is trivial if the mixture has a unique global optimum. If the mixture density has multiple global optima, then each global optimum $x_j$ has a ball $B_{j,\epsilon}$ centered around it and the samples collected will concentrate in  $B_{j,\epsilon}$. And each global optimum can be found. 
\end{proof}

\subsection{Connection with the Equi-Energy sampler}

Although our algorithm shares several similarities with
 (and benefits of) the Equi-Energy sampler, there are also several differences.
 First, our algorithm is conceptually simpler, effectively because it
 dispenses with the energy ring
 and energy cut-off of the Equi-Energy sampler. This makes it easier to
 understand, and perhaps also slightly
 easier to implement, which we believe could facilitate
 its more widespread use (despite its appeals, the Equi-Energy sampler appears
 not to have been used very extensively in practical applications).
 Second, from a more technical
 perspective, our algorithm makes
 use of a small-world (fast converging) chain as the exploring chain.  This has two advantages. First, Theorem \ref{thm:explore} gives us insights on why the tempering helps, while such a theorem does not hold for a local chain used in the Equi-Energy sampler.  Second, it has been argued \citep{jarner.roberts.07} that, if the target distribution $\pi$ is heavy tailed, the proposal distributions should have heavy tails as well. When combining high temperature with energy cutoff, the target distribution is likely to behave like a heavy tailed distribution, which will be a challenge for a local chain, while not much so for a Small-World chain.
 In addition, STEEP provides a different perspective in that every chain in STEEP is a Small-World chain, and a hot chain provides informed long-range proposals for the adjacent cold chain.  However, in practice, at least for the simple examples given in the end of the paper, we would not expect much difference in performance between STEEP and the Equi-Energy sampler. 

\def\tY{\widetilde{Y}}
\def\ty{\tilde{y}}
\def\mF{\mathcal{F}}
\def\mf{\mathcal{F}}
\def\mX{\mathcal{X}}
\def\mY{\mathcal{Y}}


\section{Ergodicity} \label{sec:ergodic}
We devote this section to state and prove ergodicity, Theorem \ref{thm:ergodic}, for Algorithm \ref{alg:many} (multiple chains). 
%
We first extend a key result of \cite[][Lemma $3.1$]{atchade+rosenthal} to two chains (Algorithm \ref{alg:two}) to obtain Theorem \ref{thm:mixingale}. Then we extend a theorem in \cite[][Theorem $3.3$]{atchade+liu} to $V$-geometric  ergodic  to obtain Theorem \ref{thm:converge}. Our proof follow closely to theirs. There are two technical lemmas: Lemma \ref{alpha} contributes to prove Theorem \ref{thm:mixingale}; and Lemma \ref{lem:pointset} helps Theorem \ref{thm:mixingale} to prove Theorem \ref{thm:converge}. We then state and prove Theorem \ref{thm:ergodic}.  We conclude this section by discussing convergence rate of Algorithms \ref{alg:two} and \ref{alg:many}.

\medskip

In Algorithm \ref{alg:two} $\{Y_n\}$ is the exploring chain and $\{X_n\}$ is the sampling chain.  Let $X_0=x_0$, $Y_0=y_0$.
Let $\{E_n\}$ denote a sequence of operators that generates $\{Y_n\}$, where $E_i$ may be identical. 
Denote $E_{1:n} = E_1 \cdots E_n$. Recall $\xi_n$ is an empirical measure generated by process $\{Y_n\}$. 
 Let $\{P_{\xi_n}\}$ denote the sequence of operators that generates $\{X_n\}$, where
 \begin{equation} \label{def:pxi}
P_{\xi_n}(x, A)  = (1-s) T(x, A) + s K_{\xi_n}(x, A),
\end{equation}
where $T$ is a local chain with stationary distribution $\pi$ and
\begin{equation}
K_{\xi_n} f(x) = \frac{1}{n}\sum_{j=1}^{n}{\alpha(x, y_j) f(y_j)} + \frac{1}{n} f(x) \sum_{j=1}^{n}{(1-\alpha(x,y_j))},
\end{equation}
where $y_i$'s are samples that induce empirical measure $\xi_n$, and 
$\alpha(x,y) = \min\left(1, \frac{\pi(y)}{\pi(x)} \frac{\pi_t(x)}{\pi_t(y)} \right)$. For a sequence of operators $\{P_n\}$, denote $P_{i:j} = P_i P_{i+1} \cdots P_j$ for $i< j$, $P_{i:i}=P_i$, and $P_{i:j} = I$ if $i > j$.
 Let $\mF_n^X=\sigma(X_1, \dots, X_n)$ and $\mF_n^Y=\sigma(Y_1, \dots, Y_n)$ be the filtrations of process $\{X_n\}$ and $\{Y_n\}$ respectively, and
let $\mF_n = \sigma(X_1, \dots, X_n, Y_1, \dots, Y_n)$.

The following Lemma is needed to prove Theorem \ref{thm:mixingale}.
 \begin{lemma} \label{alpha}
Define $\alpha_n := ||\xi_{n} - \xi_{n-1}||_V$, 
and assume
there exist constants $M<\infty$ and $\rho \in (0,1)$, such that for each $x \in \mX$
$||P_{\xi_n}^j(x, \cdot) - \pi_{\xi_n}(\cdot)||_V \le M \rho^j V(x)$ $P_{y_0}$-a.s., 
then
\begin{enumerate}
\item $\alpha_n$ is $\mF_n^Y$ measurable and $\E_{y_0}(\alpha_n) \le O({1}/{n})$.
\item $||P_{\xi_{n}}(x, \cdot) - P_{\xi_{n-1}}(x, \cdot)||_V < 2s\; V(x) \cdot \alpha_n$.
\item $||\pi_{\xi_{n}} - \pi_{\xi_{n-1}}||_V \le M \alpha_n$ for some constant $M$.
\end{enumerate}
\end{lemma}

\begin{proof}
$\alpha_n = \sup_{|f|<V} |\xi_{n}(f) - \xi_{n-1}(f)| \le \frac{1}{n} |V(Y_{n}) + \frac{1}{n-1}\sum_{i=1}^{n-1}{V(Y_i})|$ and note that $\E_{y_0}(V(Y_n)) < \infty$ for all $n$ and that $\E_{y_0}(V(Y_n)) \rightarrow \E(V) < \infty$ as $n\rightarrow \infty$, and claim (1) follows.  
Let $|f| < V$, then
\begin{eqnarray} \label{eqn:Pdiff}
\begin{aligned}
|P_{\xi_{n}}&(x, f) - P_{\xi_{n-1}}(x,f)| \\
&= s \left| \int{\alpha(x,y) f(y) [\xi_{n} - \xi_{n-1}](dy)} + f(x) \int{[1-\alpha(x,y)] [\xi_{n} - \xi_{n-1}](dy)} \right| \\
& \le  s \left|\int{ [f(y) + f(x)] [\xi_{n} - \xi_{n-1}](dy)} \right| \\
& \le  2\;s V(x) \left|\int{V(y) [\xi_{n} - \xi_{n-1}](dy)} \right| \\
& \le  2\;s V(x) ||\xi_{n}-\xi_{n-1}||_V,  
\end{aligned}
\end{eqnarray}
and the conclusion (2) follows.
From Lemma B.1. of \citet{andrieu+etal}, we have
$$|(P_{\xi_{n}}^k - P_{\xi_{n-1}}^k)(x, f)| \le M \sup_{x\in \xx} \frac{||(P_{\xi_{n}}(x,\cdot) - P_{\xi_{n-1}}(x, \cdot)||_V}{V(x)},$$
for all $k$. Let $k\rightarrow \infty$ and combine with (2) to get desired result (3).  
\end{proof}

\begin{theorem} \label{thm:mixingale}
In the proceeding framework, and assume
\begin{enumerate}  
\item For each $x \in \mX$, $||E_{1:n}(x, \cdot) - \pi_t(\cdot)||_V \rightarrow 0$, as $n\rightarrow \infty$.
\item There exist constants $M<\infty$ and $\rho \in (0,1)$, such that for each $x \in \mX$
$||P_{\xi_n}^j(x, \cdot) - \pi_{\xi_n}(\cdot)||_V \le M \rho^j V(x)$ $P_{y_0}$-a.s.
\end{enumerate}
In addition, for finite constant $c_1, c_2$, define
\begin{equation}
B(c_1, c_2, n) =  \min_{1\le k \le n} [c_1\frac{k}{n-k} + c_2 \rho^k] .
\end{equation}
Define $g_{k, \xi_k} = f - \pi_{\xi_k}(f)$. Then for any $|f|\le V$, we have
\begin{equation} \label{eqn:mixingale1}
|\E_{(x_0, y_0)}{(g_{n+j,\; \xi_{n+j}}(X_{n+j}) \big | \mf_n)}| \le B(k_1, k_2, j)V(X_n)
\end{equation}
$P_{x_0,y_0}$ a.s.  and as an immediate consequence,
\begin{equation} \label{eqn:ergodic1}
|\E_{(x_0, y_0)}{\left(f(X_{n})-\pi_{\xi_n}(f)\right)|} \le B(k_1, k_2, n)V(x_0).
\end{equation}
In addition,
\begin{equation}\label{eqn:ergodic2}
\frac{1}{n} \sum_{i=1}^{n}{\left( f(X_i) - \pi_{\xi_i}(f)\right)} \rightarrow 0  \hspace{.1in} \mbox{ as } n \rightarrow \infty, P_{x_0, y_0}-{a.s.}   
\end{equation}

\end{theorem}
 
  \begin{proof}[Proof of Theorem \ref{thm:mixingale}]
 For notational convenience, we denote $P_{\xi_n}$ by $P_n$ in the following equation.  
 \begin{eqnarray*}
 \begin{aligned}
 &\Big{|}\E_{x_0}(g_{n,\; \xi_{n}}(X_{n+j}) | \mF_n) \Big{|} = \Big{|}P_{n:(n+j-1)}f(X_{n}) - \pi_{\xi_n}(f)\Big{|} \\
 &\le \Big{|}\sum_{k=1}^{j-1}(P_n^k - \pi_{\xi_n})(P_{n+k}-P_n)P_{(n+k+1):(n+j-1)} f(X_n)\Big{|} + \Big{|}P_n^jf(X_n) - \pi_{\xi_n}(f) \Big{|}\\
 & \le \Big{|}\sum_{k=1}^{j-1}M_1 (P_n^k - \pi_{\xi_n}) (P_{n+k}-P_n)  V(X_n)\Big{|} + M_2 \rho^j V(X_n) \\
 & \le \sum_{k=1}^{j-1}M_3 k \alpha_n  (P_n^k - \pi_{\xi_n}) V(X_n) + M_2 \rho^j V(X_n)\\
 & \le M_4 \alpha_n V(X_n) \sum_{k=1}^{j-1} \rho^k k   + M_2 \rho^j V(X_n) \\
 & \le [M_4 \frac{\alpha_n}{(1-\rho)^2} + M_2 \rho^j] V(X_n),
 \end{aligned}
\end{eqnarray*}
where $\alpha_n$ is defined in Lemma \ref{alpha} (1). The transition between the first and the second line comes from the well known identity 
$P_{1:n} = \sum_{k=1}^{n-1}{P_{1}^{k}(P_{k+1}-P_{1})P_{(k+2):n}} + P_1^n$;
in the transition between the second and the third line, we recursively apply drift condition (A3); and in the transition between the third and the fourth line we applied Lemma \ref{alpha} (2) with telescoping sum. Now taking into account of Lemma \ref{alpha} (3) and again with telescoping sum, we have
\begin{eqnarray}
 \begin{aligned}
& \Big{|}\E_{x_0}(g_{n+j,\; \xi_{n+j}}(X_{n+j})|\mF_n)\Big{|} <   [M_4 \frac{\alpha_n}{(1-\rho)^2} + M_2 \rho^j] V(X_n) + j M \alpha_n \\
& <  [M_5 \alpha_n\;j + M_2 \rho^j] V(X_n),  
 \end{aligned}
\end{eqnarray}
where $M_5 = 2 \max(M_4/(1-\rho)^2, M)$.
Using the filtration trick in the end of Lemma $3.1$ of \citet{atchade+rosenthal}, we get for $k=1, \dots, j$
\begin{eqnarray*}
 \begin{aligned}
& \Big{|}\E_{x_0} (g_{n+j,\; \xi_{n+j}} (X_{n+j})|\mF_n)\Big{|} 
=\Big{|}\E_{x_0}\big{[} \E_{x_0} \Big{(}g_{n+j,\; \xi_{n+j}}(X_{n+j})\big{|}\mF_{n+j-k} \Big{)} \big{|} \mF_n \big{]} \Big{|} \\
&\le \E_{x_0}\big{[}\Big{|} \E_{x_0} \Big{(}g_{n+j,\; \xi_{n+j}}(X_{n+j})\big{|}\mF_{n+j-k} \Big{)}\Big{|} \big{|} \mF_n \big{]} \\
&\le  \min_{1\le k \le j}{[M_5 \alpha_{n+j-k}\;k + M_2 \rho^k]} \E_{x_0}(V(X_n+j-k)|\mF_n) \\
&\le \min_{1\le k \le j}{[M_5 \alpha_{n+j-k}\;k + M_2 \rho^k]}  V(X_n).
 \end{aligned}
\end{eqnarray*}
Taking expectation with respect to process $\{Y_n\}$, we get
\begin{eqnarray} \label{eqn:mixingale2}
 \begin{aligned}
& |\E_{x_0,y_0} (g_{n+j,\; \xi_{n+j}}(X_{n+j})|\mF_n)|  
 \le  \E_{y_0}\{\min_{1\le k \le j}{[M_5 \alpha_{n+j-k}\;k + M_2 \rho^j]}\} V(X_n)  \\
& \le \min_{1\le k \le j} [M_6 \frac{k}{n+j-k} + M_2 \rho^j] V(X_n),  
 \end{aligned}
\end{eqnarray}
which is claim (\ref{eqn:mixingale1}).
Taking $n=0$ in (\ref{eqn:mixingale2}), we obtain (\ref{eqn:ergodic1}).

Note $B(c_1, c_2, n) \rightarrow 0$ as $n\rightarrow \infty$ (see Section \ref{sec:rate_steep}), this and (\ref{eqn:mixingale2}) show that
\begin{equation} \label{eqn:ergodic3}
\E_{x_0,y_0}{\left(f(X_n) - \pi_{\xi_n}(f)\right)} \rightarrow 0, \mbox{    as } n\rightarrow \infty.
\end{equation}
Following argument similar to the proof of Theorem $3.2$ of \citet{atchade+rosenthal},
write $Z_n = g_{n,\xi_n}(X_n) - \E_{x_0,y_0}{(g_{n,\xi_n}(X_n))}$. Given (\ref{eqn:mixingale2}), $(Z_n, \mF_n)$ is an $L^2$-mixingale with mixingale sequence $c_n$ being a constant and $\psi_n = B(c_1, c_2, n)$.   Then Corollary $2.1$ of \citet{davidson+dejong} implies that
\begin{equation}\label{eqn:ergodic4}
\frac{1}{n} \sum_{k=1}^{n}{\left( g_{k,\xi_k}(X_k) - \E_{x_0,y_0}{g_{k,\xi_k}(X_k)}\right)} \rightarrow 0, \hspace{.1in} P_{x_0,y_0} \mbox{ a.s. as } n \rightarrow \infty.
\end{equation}
Combine (\ref{eqn:ergodic3}) and (\ref{eqn:ergodic4}) we obtain (\ref{eqn:ergodic2}).
\end{proof}

The following lemma is needed to prove the Theorem \ref{thm:converge}.
It is a technical lemma that is an extension of of Lemma $3.1$ in \citet{atchade+liu}.

\begin{lemma} \label{lem:pointset}
Let $\{f_n\}$ be a sequence of measurable functions and let $\{\mu_n\}$ be a sequence of probability measures such that $|f_n| < V$ and $f_n \rightarrow f $ pointwise and $\mu_n \rightarrow \mu$ setwise. In addition, $\int{V(x) \mu_n(dx)} < \infty$ and  $\int{V(x) \mu(dx)} < \infty$. Then
$\int{f_n(x) \mu_n(dx)} \rightarrow \int{f(x) \mu(dx)}$.
\end{lemma}
\begin{proof}[Proof of Lemma \ref{lem:pointset}]
In light of Proposition $18$ in Chapter $11$ of \citet{roydenbook}, it is sufficient to prove that $\int{V(x) \mu_n(dx)} \rightarrow \int{V(x) \mu(dx)}$. However, setwise convergence of $\mu_n \rightarrow \mu$ implies that for any simple function $\psi$, we have $\int{\psi \mu_n(dx)} \rightarrow \int{\psi \mu(dx)}$. Since simple functions are dense (in $L^p$), we have two sequences of simple functions $2 V > \psi_m \ge V$ and $\phi_m \le V$ such that $\psi_m \searrow V$ and $\phi_m \nearrow V$, and $\int{\psi_m \mu_n(dx)} \ge \int{V \mu_n(dx)}
\ge \int{\phi_m \mu_n(dx)}$. Let $n \rightarrow \infty$, we have $\int{\psi_m \mu(dx)}
\ge \lim_n{\int{V \mu_n(dx)}} \ge \int{\phi_m \mu(dx)}$ for any $m$.  Now let $m \rightarrow \infty$ and apply Lebesgue's dominant convergence theorem to finish the proof.
\end{proof}

\begin{theorem} \label{thm:converge}
With the setting outlined as in Theorem \ref{thm:mixingale}, assume:
\begin{enumerate}
\item For any $x \in \xx$ and $A \in \mathcal{B}$, $P_{\xi_n}(x, A) \rightarrow P_{\xi}(x, A)$ as $n\rightarrow \infty$, $P_{y_0}$-a.s.
\item
There exists a finite constant $M$ and a $\rho \in (0,1)$ such that $||P_{\xi}^k(x, \cdot) - \pi(\cdot)||_V < M V(x) \rho^k$ and $||P_{\xi_n}^k(x, \cdot) - \pi_{\xi_n}(\cdot)||_V < M V(x) \rho^k$ $P_{y_0}$-a.s.
\end{enumerate}
Then for any measurable function $f: \mX \rightarrow \mathcal{R}$ such that  $|f| < V$ we have
\begin{equation}
\E_{x_0, y_0}[f(X_n)] \rightarrow \pi(f) \hspace{.2in} and \hspace{.2in} \frac{1}{n} \sum_{k=1}^n{f(X_k)} \rightarrow \pi(f) \hspace{.1in} P_{x_0,y_0}-a.s.
\end{equation}
\end{theorem}

\begin{proof}[Proof of Theorem \ref{thm:converge}]
For any $|f|<V$, by Theorem \ref{thm:mixingale} we have $\E_{x_0,y_0}[f(X_n) - \pi_{\xi_n}(f)] \rightarrow 0$ and $\frac{1}{n} \sum_{k=1}^n{[f(X_k) - \pi_{\xi_k}(f)]} \rightarrow 0$ $P_{x_0,y_0}$-a.s. as $n \rightarrow \infty$. To finish we need to prove that $\pi_{\xi_n}(f) \rightarrow \pi(f)$ $P_{y_0}$-a.s. as $n \rightarrow \infty$.  By assumption we have $P_{\xi_n}f(x) \rightarrow P_\xi f(x)$ $P_{y_0}$-a.s. for all $x \in \xx$. By Lemma \ref{lem:pointset}, $P_{\xi_n}^2f(x) = P_{\xi_n}(P_{\xi_n}f(x)) \rightarrow P_\xi^2f(x)$. By recursion, for any $k\ge1$,
\begin{equation} \label{eqn:p_xi}
P_{\xi_n}^kf(x) \rightarrow P_\xi^kf(x) \hspace{.1in} \mbox{ $P_{y_0}$-a.s.  as } n \rightarrow \infty.
\end{equation}
We have
\begin{eqnarray} \label{eqn:pi_xi}
\begin{aligned}
& |\pi_{\xi_n}(f) - \pi(f)| 
 \le  |\pi_{\xi_n}(f) - P_{\xi_n}^kf(x)| + |P_\xi^kf(x) - \pi(f)| 
 +  |P_{\xi_n}^kf(x) - P_\xi^kf(x)| \\ 
 & \le  2 M V(x) \rho^k + |P_{\xi_n}^kf(x) - P_\xi^kf(x)|.
\end{aligned}
\end{eqnarray}
Combine above with (\ref{eqn:p_xi}) we have $|\pi_{\xi_n}(f) - \pi(f)| \rightarrow 0$ $P_{y_0}$-a.s.
\end{proof}

Recall in Algorithm \ref{alg:many}, the $H$-th chain is a homogeneous chain. Conditioning on the realization of $\{X^{(i+1)}_n\}$ (for $i < H$),  $\{X_n^\ii\}$ is a nonhomogeneous Markov chain with transition kernel $P_n^\ii$.  
The $P_n^\ii$ operates on $f$ such that:
$P_n^\ii f(x) = (1-s) T^\ii f(x) + s K_{\xi_n^{(i+1)}}^\ii f(x)$,
where $T^\ii$ is a homogeneous local chain with stationary distribution $\pi_i$, and
\begin{equation} \label{eqn:kernel:kn}
K_{\xi_n^{(i+1)}}^\ii f(x) = \frac{1}{n}\sum_{j=1}^{n}{\alpha(x, y_j) f(y_j)} + \frac{1}{n} f(x) \sum_{j=1}^{n}{(1-\alpha(x,y_j))},
\end{equation}
where $y_i$'s are samples that induce empirical measure $\xi_n^{(i+1)}$, and
$ \alpha(x,y) = \min\left(1, \frac{\pi_i(y)}{\pi_i(x)} \frac{\pi_{i+1}(x)}{\pi_{i+1}(y)} \right)$.
Let $P^\ii(x, A)=(1-s) T^\ii f(x) + s K_{\xi^{(i+1)}}^\ii f(x)$ be the limiting transition kernel as $n\rightarrow\infty$ where
$$K_{\xi^{(i+1)}}^\ii f(x) = \int_{\xx}{\alpha(x, y) f(y)\pi_{i+1}(dy)} + f(x) \int_{\xx}{\left(1-\alpha(x,y)\right) \pi_{i+1}(dy)}.$$ We have the following ergodic theorem.
\begin{theorem} \label{thm:ergodic}
In proceeding framework, assume $T^\ii, i \in \{0, \dots, H\},$ satisfies assumptions (A1-A3), then for $|f| < V$, as $n \rightarrow \infty$,
\begin{equation}
\E[f(X_n^\ii)] \rightarrow \pi_i(f) \hspace{.2in} and \hspace{.2in} \frac{1}{n} \sum_{k=1}^n{f(X_k^\ii)} \rightarrow \pi_i(f) \hspace{.2in} a.s.
\end{equation}
\end{theorem}

\begin{proof}[{Proof of Theorem \ref{thm:ergodic}}]
Let $x_0^{(i)}$ be the starting point of process $\{X_n^{(i)}\}.$
It is easy to check the detailed blance equation holds for the process $\{X_n^{(H)}\}$, so the claim holds for chain $H$.

For each $i < H$, we want to show that the $i$-th chain in STEEP algorithm satisfies two assumptions in Theorem \ref{thm:converge}. 

Since process $\{X_n^{(i+1)}\}$ has stationary distribution $\pi_{i+1}$, we have for all $x \in \xx$,  $P_{n}^{(i+1)}(x, \cdot) \rightarrow \pi_{i+1}.$
Since $\alpha(x,y)$ as in (\ref{eqn:kernel:kn}) is bounded, by Lebesgue's dominate convergence theorem, $K_{\xi_n^{(i+1)}}^\ii f(x) \rightarrow K_{\xi^{(i+1)}}^\ii f(x)$, $P_{x_0^{(i+1)}}$-a.s.,  which implies $P_n^\ii f(x) \rightarrow (1-s) T^\ii f(x) + s K_{\xi^{(i+1)}}^\ii f(x)$ $P_{x_0^{(i+1)}}$-a.s.,
and hence $P_n^\ii(x, A) \rightarrow P^\ii(x,A)$ $P_{x_0^{(i+1)}}$-a.s. for each $x \in \xx$ and $A \in \mathcal{B}$.  So the Assumption (1) hold true.

The drift and minorization conditions on $T^\ii$ transfers to $P_n^\ii$, so that each $P_n^\ii$ admits an invariant distribution $\pi_{i,n}$ $P_{x_0^{(i+1)}}$-a.s. and is geometric ergodic. The limiting transition kernel $P^\ii$ also inherit the drift and minorization conditions on $T^\ii$ and admits invariant distribution $\pi_i$ $P_{x_0^{(i+1)}}$-a.s. So the Assumption (2) hold true.

By induction, Theorem \ref{thm:ergodic} now follows Theorem \ref{thm:converge}.
\end{proof}

\subsection{Convergence Rate} \label{sec:rate_steep}
The convergence of Algorithm \ref{alg:two} was broken down into two parts. The first part is $(P_{\xi_1} \cdots P_{\xi_n})(x, \cdot) \rightarrow \pi_{\xi_n}$ (Theorem \ref{thm:mixingale}), which is the convergence of the sampling chain;  and the second part is $\pi_{\xi_n} \rightarrow \pi_t$ (Theorem \ref{thm:converge}), which is the convergence of the exploring chain.   We shall discuss them separately.

First note the exploring chain in Theorem \ref{thm:mixingale} is more general in that it can be either a homogenous or a non-homogenous chain. 
From (\ref{eqn:pi_xi}) and Lemma \ref{alpha}
\begin{eqnarray} \label{eqn:conv_rate}
\begin{aligned}
|\pi_{\xi_n}(f) - \pi(f)|  \le & 2 M V(x) \rho^k + |P_{\xi_n}^kf(x) - P_\xi^kf(x)| \\
 \le & 2M V(x) \rho^k + 2sV(x)\|\xi_n-\xi\|_V \\
 \le & (c_1 \rho^k + c_2 \|\xi_n - \xi\|_V) V(x), 
\end{aligned}
\end{eqnarray}
where for any $|f| < V$, we have $\xi(f) = \pi_t(f)$ $P_{y_0}$-a.s. 
Clearly, the convergence rate of 
$\pi_{\xi_n} \rightarrow \pi$ is dominated by the rate $c_2\|\xi_n-\xi\|_V$, which depends on the convergence rate of the exploring chain. 
  
In light of the discussion in \citet{atchade+rosenthal}, Theorem \ref{thm:mixingale} implies that the sampling chain in Algorithm \ref{alg:two} converges (to $\pi_{\xi_n}$) at rate
\begin{equation}\label{eqn:B}
B(n, \rho) = \min_{1\le k \le n}{(c_1 k/(n-k) + c_2 \rho^k)}.
\end{equation}
Take derivative with respect to $k$ and set to $0$ to get
\begin{equation} \label{eqn:nk}
c_1 \frac{n}{(n-k)^2} = c_2 \log{\frac{1}{\rho}} \;\rho^k.
\end{equation}
If we assume $k =O(n)$, then (\ref{eqn:nk}) reduce to $\frac{1}{n} \approx c \rho^k$, solve to get $k = O(\log{(n)})$, and we reach contradiction.  So we may assume $k=o(n)$, and (\ref{eqn:nk}) simplifies to
$c_1 \frac{1}{n} \approx c_2 \log{\frac{1}{\rho}} \;\rho^k.$
Take log on both sides and solve to get
 $k \approx -\log{n}/\log{\rho}.$
 Substitute back to (\ref{eqn:B}) and use $\log{(1-x)} \approx -x$ when $x$ small and note $\rho^k \approx 1/n$  we get
\begin{equation} \label{eqn:brate}
B(n,\rho) \approx O(\frac{n^{-1} \log{n}}{1-\rho}).
\end{equation}
Loosely, $(1-\rho) \approx Gap(P_\xi)$, where $P_\xi$ is the limiting transition kernel of the sampling chain.  From (\ref{eqn:brate})  the sampling chain of Algorithm \ref{alg:two} converges at a polynomial rate that is also proportional to $1/\Gap(P_\xi)$. Note the rate is not a function of the \emph{size} of the $\pi(x)$.

Extending to Algorithm \ref{alg:many}, the above analysis suggests that the sampling chains converge at rate $O(\tau^{\lambda d} \; n^{-1}\log{n})$ for some $\lambda \in [1,2]$ (Theorem \ref{thm:sample}), where $\tau>1$, is the ratio between adjacent temperatures. A large $\tau$ slows down the convergence of the sampling chains. However, the spectral gap of the exploring chain is enlarged by a factor of $\tau^{Hd}$ (Theorem \ref{thm:explore}). Hard problems will benefit from such a trade-off because their exploring chains usually have very small spectral gaps.  

%


\section{Spectral Gaps} \label{sec:gap}
Our main aim in this section is to prove Theorems \ref{thm:explore} and \ref{thm:sample}.
We rely on the state decomposition theorem (Theorem \ref{thm:SDT}) to analyze the Algorithm \ref{alg:two}: we partition a multi-modal distribution into pieces log-concave pieces and analyze each restricted local chain $P_{A_i}$ and the component chain $P_c$ separately.
The $\Gap(P_{A_i})$ is well studied \citep{LV03,peter,rudolf,guankrone} using the isoperimetric inequality \citep{LS,KLS}, and $\Gap(P_c)$ can be computed directly using conductance and the Cheeger's inequality.

\subsection{Conductance and spectral gap}
 
Recall $P(x, dy)$ defined in (\ref{eqn:P}),
for $A \in \mathcal{B}$ with $\pi(A)>0$, define
\begin{equation} \label{eqn:flow}
\h_P(A)=\frac{1}{\pi(A)}\int_{A}{P(x, A^c)\pi(dx)}.
\end{equation}
The quantity $\h_P(A)$ can be thought of as the (probability) flow
out of the set $A$ in one step when the Markov chain is at
stationarity.
The \textit{conductance} of the chain is defined by
\begin{equation}
\h_P=\inf_{ 0<\pi(A)\le 1/2}{\h_P(A)}. \label{eqn:conductance}
\end{equation}
Intuitively, a small $\h_P$  implies mixing slowly because  the chain may be trapped in a set whose
measure is $\le 1/2$. On the other hand, a
large $\h_P$ implies mixing rapidly as  nowhere is sticky.
We have the following theorem \citep{lawlersokal}.  
\begin{theorem}[Cheeger's Inequality] \label{Cheeger}
Let P be a reversible Markov transition kernel with invariant
measure $\pi$. Then
\begin{equation}\label{eqn:cheeger}
\frac{\h_P^2}{2} \le \Gap (P) \le 2\:\h_P.
\end{equation}
\end{theorem}

Let
$k(x,dy)=(1-s)k_1(x,dy)+s\: k_2(x,dy),$ with $0\leq s \leq 1$.
Suppose operators $P,$ $P_1,$ and $P_2$ are induced by $k, k_1, k_2$ respectively. Clearly,
\begin{equation} \label{eqn:Pplus}
P=(1-s) P_1+s\: P_2.
\end{equation}
It is straightforward to show the following Lemma \citep{guankrone}, which allows one to bound the spectral gap for a mixture of kernels.
\begin{lemma} \label{lem:gap}
Suppose $P$ is defined by
(\ref{eqn:Pplus}).  Then
\begin{equation} \label{eqn:concave}
\Gap(P)\geq \frac{1}{2} \max{\left((1-s)^2\:\h_{P_1}^2, s^2 \h_{P_2}^2\right)};
\end{equation}
\end{lemma}

The following theorem was proved in \citet{guankrone}.
\begin{theorem} \label{lowerbound}  
Suppose $\pi(x)$ is an $\alpha$-smooth log-concave probability density of $d$-dimension on a convex set
$K$. Suppose further that $\pi$ has barycenter $0$ and set
$M_{\pi}=\int_{K}{|x| \: \pi(dx)}$. Then the conductance, $\h_P$,
of the Metropolis-Hastings chain with transition kernel $P(x,dy)$
induced by the uniform $\delta$-ball proposal satisfies
\begin{equation} \label{eqn:hp}
\h_P \ge \frac{\delta\:e^{-\alpha\:\delta}}{1024\:\sqrt{d}\: M_\pi}.
\end{equation}
\end{theorem}

 Combining Equations \eqref{eqn:concave} and \eqref{eqn:hp} we obtain a lower bound of the spectral gap of a local chain sampling a log-concave distribution.  Set $\delta = 1/\alpha$ to see that it is fast-mixing. This and the state decomposition theorem leads to a conclusion that a Small-World chain is fast mixing \citep{guankrone}.  Moreover, because the dimensionality $d$ is in a polynomial form in (\ref{eqn:hp}), so there is no ``curse of dimensionality" with a local chain sampling a log-concave distribution. On the contrary, there is a ``curse of dimensionality" in the component chain simply because ``volumes" of modes matter \citep{guankrone}.  This justifies our focus on the component chain in Theorems \ref{thm:explore} and \ref{thm:sample}.

\subsection{Proof of The Thorems \ref{thm:explore} and \ref{thm:sample}}
We need to establish the connection between a log-concave distribution and its powered-up alternatives.
\begin{lemma} \label{lem:1}
Let $f(x)$ be a log-concave distribution on $\xx$ of dimension $d$,  and $f(x)$ attain its unique maximum  at $0$. Let $f_t(x) \propto f(x)^{1/t}$, where $t > 1$. Then (a) $\frac{f_t(0)}{f(0)} \ge t^{-d}$, (b) $\frac{f_t(x)}{f(x)} \ge \frac{f_t(0)}{f(0)}$ for  any $x \in \xx$, and (c) if $f(x)$ is $\alpha$-smooth, then the equality in (a) and (b) holds up to a constant.  
\end{lemma}

\begin{proof}
By definition of log-concave, for $0 \leq s \leq 1$,
$$f(s\; x + (1-s) y) \ge f(x)^{s} f(y)^{1-s}.$$ Let $s = 1/t$ and $y=0$. Choose $h$, such that $1/t+1/h=1$, then
\begin{equation} \label{eqn:f}
f(x/t) \ge f(x)^{1/t} f(0)^{1/h}.
\end{equation}
Since $\int_{\xx}{f(x)dx}=1$, we can get 
$\int_{\xx}{f(x/t)dx}  = t^d \int_{\xx}{f(x/t) d(x/t)} = t^d.$ So 
$$\int_{\xx}{f(x)^{1/t} f(0)^{1/h} dx} = f(0)^{1/h} \int_{\xx}{f(x)^{1/t} dx} \leq t^d,$$
 rearrange term to get:
 $\frac{f(0)^{1/t}} {\int_{\xx }{f(t)^{1/t} dx}} \ge f(0) \; t^{-d}.$
Identify the left hand side is $f_t(0)$ to prove (a).

Since $f(x)$ is log-concave, for any unit vector $u \in \xx$, $f(uz)$ is monotone decreasing in $z$, hence $f(z u)^{1/t -1}$ is monotone increasing in $z$ (as we assume $t > 1$). This proves (b).

Because of logconcavity, for any unit vector $u \in \xx$, there exists scalars $g, \nu$, such that $f(u z) < e^{-\nu z}$ for $z > g$. By $\alpha$-smooth, we have $f(u z) > e^{-\alpha z}$. For exponential functions, equality in (\ref{eqn:f}) holds and this proves (c).  
\end{proof}
\begin{remark}
Note the bound in $(a)$ is essentially tight for a general log-concave distribution, as can be seen clearly from one-dimensional exponential distributions. For distributions such as (multi-variate) normal, the bound in $(a)$ is not tight.  However, we are confined by the technical condition of $\alpha$-smoothness, and we do not pursue a better bound. Although the $\alpha$-smooth is a technical condition for the convenience to bound conductance \citep{guankrone}, it is worthwhile to note that it is crucial for the tempering scheme as well. Taking an extreme example, the tempering will not help for two uniform distributions on two unit discs that are far apart. While tempering is helpful for two normal distributions that are far apart, as we will see in Section \ref{sec:app}.     
\end{remark}
\begin{lemma} \label{lem:2}
For a piece-wise $\alpha$-smooth log-concave distributions defined in (\ref{eqn:pi}),  let $\xi_i(x) = w_i^{1/t}\pi_i(x)^{1/t} / I_i$, where $I_i=w_i^{1/t}\int_{A_i}{\pi_i(x)^{1/t} dx}$ for $i= 1, \dots, m$.  Let $I=\sum{I_i}/m$ and $\xi'_i(x) = w_i^{1/t}\pi_i(x)^{1/t} / I$ for each $i$. Then there exists constants $c_1, c_2$ such that $c_1 < \xi'_i(A_i) / \xi_i(A_i) < c_2$.  
\end{lemma}
\begin{proof}
For any pair $w_i \ge w_j$, for $t \ge 1$, we have $w_j/w_i \le w_i^{1/t}/w_j^{1/t} \le w_i/w_j$. So that without loss of generality, we may assume $w_i=1$ for  $i= 1, \dots, m$.  By (c) of Lemma \ref{lem:1}, we have $\xi_i(A_i) = c_i \pi_i(A_i) t^{-d} / I_i$. Since $\xi_i(A_i) = \pi_i(A_i) = 1$, we have $I_i=c_i t^{-d}$, which implies that $\xi'_i(A_i) / \xi_i(A_i)  =  m\: c_i / \sum{c_i}$ and the conclusion follows.
\end{proof}
\begin{remark}
Lemma \ref{lem:2} says two different normalizations, namely, normalization within each pieces and normalization combining all pieces, are equivalent up to a constant for a piece-wise $\alpha$-smooth log-concave distribution. So we can use the normalization within each pieces to ease the presentation.
\end{remark}

With Lemma \ref{lem:1} at hand, the proof of Theorem \ref{thm:sample} is easy.
 
\begin{proof}[Proof of Theorem \ref{thm:sample}]
Consider $\pi_{A_1}$,  $\pi_{A_2}$ as the $\pi$ restricted on $A_1$ and $A_2$ respectively, which we denote by $\pi_1$ and $\pi_2$. Let $\xi_1, \xi_2$ be their normalized powered-up alternatives respectively.  
The conductance bound of $P_c$ requires estimate of the integral appeared in (\ref{eqn:ph})
\begin{eqnarray} \label{eqn:prove:thm:sample}
 &\int_{A_1\times A_2}{\min\left(1, \frac{\pi_2(y)}{\pi_1(x)} \frac{\xi_1(x)}{\xi_2(y)}\right)} \pi_1(x) \xi_2(y) dx dy\\
&=\int_{A_1\times A_2}{\min\left(\frac{\xi_1(x)}{\pi_1(x)},\frac{\xi_2(y)}{\pi_2(y)}\right)} \pi_1(x) \pi_2(y) dx dy \\
&=  c_1 \frac{1}{t^d} \pi_1(A_1) \pi_2(A_2),
\end{eqnarray}
where the last equality obtained from Lemma \ref{lem:1}.  
Therefore, from Equation (\ref{eqn:conductance}) we get for some constant $c$  
\begin{equation}
h_{12} = \frac{1}{2 \pi_1(A_1)} \int_{A_1\times A_2}{\min\left(1, \frac{\pi_2(y)}{\pi_1(x)} \frac{\xi_1(x)}{\xi_2(y)}\right)} \pi_1(x) \xi_2(y) dx
 = \frac{c}{t^d}.  
\end{equation}
Hence, the conductance of the component chain $P_c$ is proportional to $t^{-d}$. Following the Cheeger's Inequality (\ref{eqn:cheeger}) to get the bound on spectral gap.   
It is clear that the bound is \emph{not} a function of the ``size" of $\pi$.
\end{proof}

\begin{proof}[Proof of Theorem \ref{thm:explore}]
Use same notations defined in the proof of Theorem \ref{thm:sample}. Let $h(x,y)$ be the long range proposal.
We have
\begin{eqnarray}
\begin{aligned}
h_{12} &= \frac{1}{2 \pi_1(A_1)}\int_{A_1\times A_2}{\min{\left(1, \frac{\pi_2(y)}{\pi_1(x)}\frac{h(y,x)}{h(x,y)} \right)} \pi_1(x) h(x,y) dx dy} \\
&=\frac{1}{2 \pi_1(A_1)}\int_{A_1\times A_2}{\min{\left(\frac{h(y,x)}{\pi_1(x)}, \frac{h(x,y)}{\pi_2(y)} \right)} \pi_1(x) \pi_2(y)dx dy} \\
&> a_1\frac{1}{2 \pi_1(A_1)} \int_{A_1\times A_2}{\pi_1(x) \pi_2(y)dx dy} \\
& = \frac{1}{2}a_1 \pi_2(A_2)
\end{aligned}
\end{eqnarray}
where
\begin{equation}
a_1 = \inf_{x\in A_1, y\in A_2}{\min{\left(\frac{h(y, x)}{\pi_1(x)}, \frac{h(x, y)}{\pi_2(y)} \right)}}.
\end{equation}

Follow same argument to get for powered up distributions $\xi_i$'s.
\begin{eqnarray}
\begin{aligned}
h_{12}^\prime &= \frac{1}{2 \xi_1(A_1)}\int_{A_1\times A_2}{\min{\left(1, \frac{\xi_2(y)}{\pi_1(x)}\frac{h(x)}{h(y)} \right)} \xi_1(x) h(y) dx dy} \\
& > \frac{1}{2}a_t  \xi_2(A_2)
\end{aligned}
\end{eqnarray}
where
\begin{equation}
a_t = \inf_{x\in A_1, y\in A_2}{\min{\left(\frac{h(y, x)}{\xi_1(x)}, \frac{h(x, y)}{\xi_2(y)} \right)}}.
\end{equation}

Note $\pi_2(A_2) = \xi_2(A_2)$, so the ratio $h_{12}^\prime / h_{12}$ is determined by $a_t / a_1$.  However, for each $x \in A_1, y \in A_2$, $\frac{h(y,x)}{\xi_1(x)} = c\;  t^{d} \frac{h(y,x)}{\pi_1(x)}$ due to Lemma \ref{lem:1} (c). So we have $a_t / a_1 = c\; t^d$ and hence
\begin{equation} \label{eqn:hprime}
h_{12}^\prime = c\; h_{12}\; t^d.
\end{equation}
For an $m \times m$ stochastic matrix $P_c = (h_{ij})$, the spectral gap can be bounded from below \citep{pena} by
\begin{equation} \label{eqn:mb}
\Gap(P_c) \ge m\; \min_{i \neq j} h_{ij}.
\end{equation}
Combine Equations (\ref{eqn:hprime}) and (\ref{eqn:mb}) to finish the proof.
\end{proof}

\section{Applications}\label{sec:app}

There are three key parameters needed to be specified in applications, namely, the proportional of the long-range proposals $s$, the number of chains $(H+1)$ and the temperature ratio $\tau$.  
The theoretically best value of $s$ is $1/3$ because it maximizes the lower bound of the spectral gap of a Small-World chain \citep{guankrone}.  Indeed, in the following two examples, we use $s = 0.33$.  Of course for a specific application one may tune $s$ based on acceptance ratio. We note, however,  $s$ should keep constant during a MCMC run.  
To specify $H$ and $\tau$ we suggest the following procedure:  First, sample $\pi_t(x)$ and tune $t$ until acceptance ratio of long-range proposal is high, say, larger than $0.2$. Next, use Algorithm \ref{alg:two} to sample $\pi_t(x), \pi_{t/\tau}(x)$, tuning $\tau$ so that the acceptance ratio for long-range proposal of the sampling chain is $>0.2$.
$H$ can be estimated by $[\log{t}/\log{\tau}]$.  One may find burn-in and thinning helpful.

Lastly, the choice of long-range proposal is often problem-dependent. If the state space is $\xx$, we recommend a heavy tailed distribution like Cauchy.  When the state space is trees or graphs, a long-range proposal is hard to define -- we suggest to compound local proposals of randomly many times to obtain a long-range proposal.   

\subsection{Sampling Needles}
In this example, our target distribution is a mixture of normals, $f(x) = 0.5 N(x; \mu_1, \Sigma_1) + 0.5 N(x; \mu_2, \Sigma_2)$ where $\mu_1 = (0,0)^t$, $\mu_2=(5,5)^t$ and $\Sigma_1 = \Sigma_2 =  \left( \begin{smallmatrix} 0.01&0\\ 0&0.01 \end{smallmatrix} \right)$.  To see the minorization and drift condition hold for this example,  see \cite{atchade10} and references therein. 

Our local proposal is two dimensional ball with radii $0.1$. The long-range proposal is two dimensional Cauchy with scale parameter $1$. Local chains are trapped in either mode within $1,000,000$ iterations (data not shown). We use $6$ chains with $\tau=6$ with $1000$ burn-in steps in each chain. The sampling chain ran $10,000$ steps, which makes the total iterations of the all $6$ chains $81,000$ steps.  Define region $A = \{x: x_1^2+x_2^2 < 0.05^2\}$ and $p = Pr(X_n^{(6)} \in A)$, the probability that the sampling chain visits $A$. Figure \ref{fig:needles} is the sample trace of a typical run, where after thinning of every $10$ steps, the last $1,000$ samples of each chain were plotted.

\begin{figure}[htp] \label{fig:needles}
\begin{center}
\includegraphics[width=8cm]{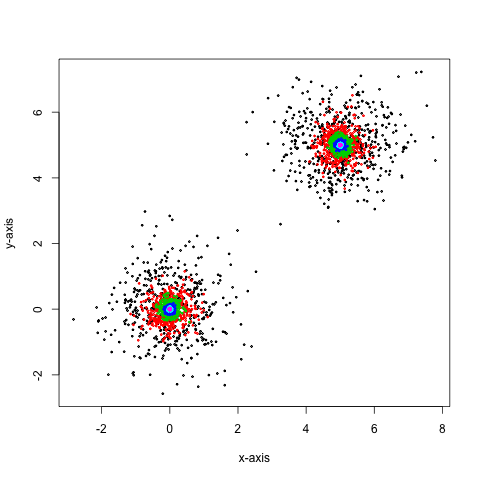}
\end{center}
\caption{Each color corresponds to different temperature $6^k$, $k=5, 4,\dots, 1, 0$.  The proportion of the sampling chain that visit region $A$ is $44.6\%$ }
\end{figure}

%

We repeat above run $100$ times and obtain statistics regarding $\hat{p}$.  The mean is $0.50$ and median is $0.49$, the standard deviation is $0.08$, and the $5$ and $95$ percentiles are $0.37$ and $0.62$ respectively. 
We repeat simulations for $\mu_2=(25,25)^t$ without modifying $l(x,\cdot)$ and $h(x,\cdot)$. With an extra chain and twice number of iterations, similar results are obtained (data not shown).

\subsection{Sampling Phylogenetic Trees}
Markov chain Monte Carlo algorithms plays an important role in (Bayesian) phylogenetic inference, perhaps through the wide-spread usage of software packages like MrBayes \citep{mrbayes} and PAML \citep{paml}. \citet{vigoda} argued that phylogenetic MCMC algorithms are misleading when data (nucleotides sequences) are generated by mixture of phylogenetic trees. Fixing branch lengths, they generated sequence data using two trees that are far apart (that is, local proposals can not reach from one to another in one step).  They first showed that there is a valley in between the two trees used to simulated sequence data, and the valley becomes steeper when the sequence length ($N$) increases. Then they argued that existing local samplers takes exponentially long iterations (in $N$) to move from one mode to another. Their theoretical results is essentially the first part of the Theorem $3.1$ in \citet{guankrone}. In light of \citet{guankrone} and theories presented in this paper, we see that the slow mixing problem presented in \citet{vigoda} can be simply resolved by a Small-World sampler, and better mixing can be achieved by a STEEP sampler.

In our simulation, we fix the branch lengths the same as those in \citep{vigoda} with inner branch lengths equal $0.1$ and tip branch lengths equal $0.01$.  We use Jukes-Canter as the evolutionary model to compute likelihood of different tree topologies. Our local proposal is the nearest neighbor interchange (NNI) \cite[c.f.][]{Fbook}, and the long-range proposals are simulated by  compounding multiple (but random many) nearest neighbor interchanges.  In this example, the minorization and drift condition hold because the state space is finite. 

The five taxa example presented in \citep{vigoda} is too simple for a simulation study. We simulate DNA sequence data \citep{seq-gen} based on two generating trees (in Newick format) $A=(((((1,2),3),4),5),6),(7,8))$ and $B=(((((1,7),3),4),5),6),(2, 8))$. Note we switch the position of taxa $2$ and $7$ and it takes $5$ NNI to move from tree A to tree B. We found it difficult to simulate sequence data from the two trees that results in similar likelihood on both, so we simulate sequences of length $1000$ from tree $A$ and switch the sequences $2$ and $7$ to obtain new sequences and concatenate two sets of sequences together. By doing so, we ensure that tree $A$ and $B$ have the same likelihood.  

We first ran a simple Metropolis-Hastings algorithm and plot the distance between taxa $1$ and $2$, the local chain were trapped within one mode during $1,000,000$ steps (data not shown). We then ran STEEP sampler of $4$ chains with $\tau=10$, each chain ran $50,000$ steps with $5,000$ steps of burn-in. After thinning every $10$ steps, the last $5000$ were plotted for each chain. The left panel shows the likelihood trace plot of each chain. The right panel shows the distance (between taxa $1$ and  $2$) trace plot. Clearly, the chain moves frequently between trees $A$ and $B$.
\begin{figure}[htp]
\begin{center}
\includegraphics[width=6cm]{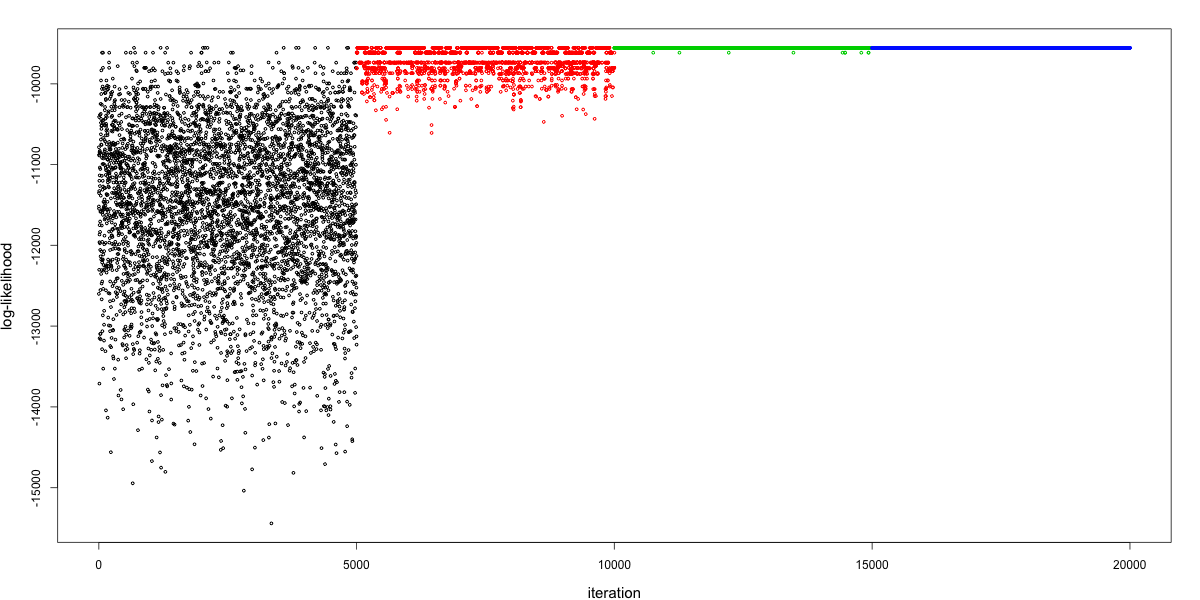}
\includegraphics[width=6cm]{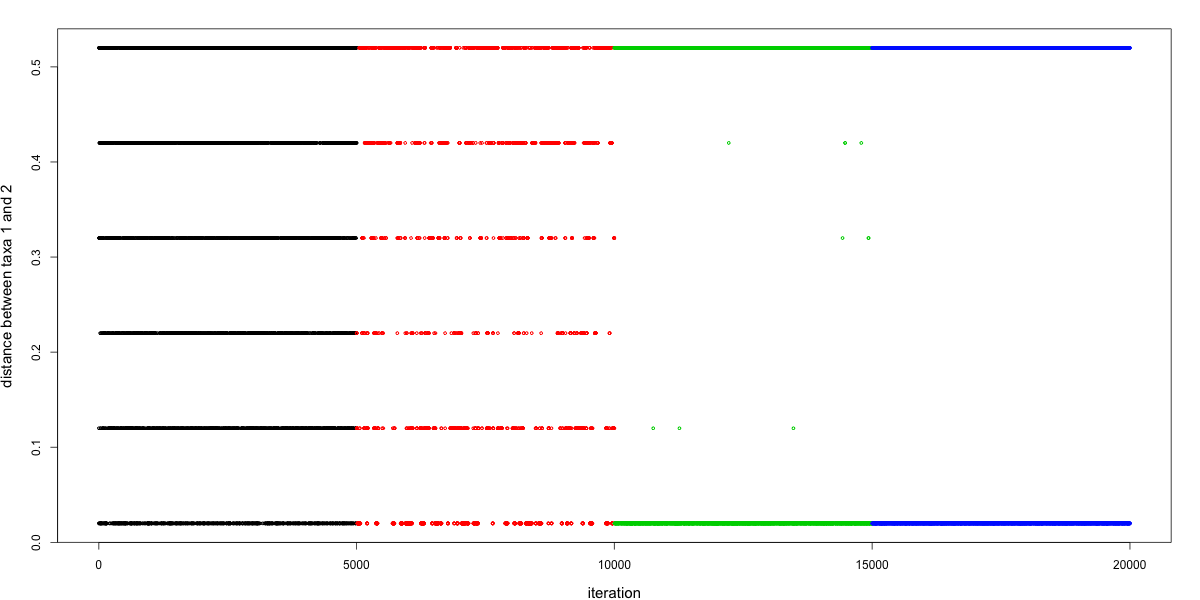}
\end{center}
\caption{Each color band corresponds to a different chain of temperature $10^k$, $k=3, 2, 1, 0$. The best log-likelihood of trees that is one local proposal (NNI) away from the generating trees is $60.44$ smaller than those of generating trees. Only two generating trees appear in the last chain and the proportion of the two generating trees are $44.0\%$ and $56.0\%$ respectively.}
\end{figure}

This toy example demonstrates that the STEEP performs well in mixture of trees where a local chain fails. Note in the example we fixed the branch lengths and evolutionary parameters.  We invite authors of MrBayes, PAML and others to further investigate its performance when taking into account of branch lengths and evolutionary models.    

\section{Discussion}

We have presented a new sampling algorithm, proved its ergodicity, and demonstrated its usefulness. The analysis of the spectral gap appears to be new. Although the theory is presented in the Euclidian state space, we believe it applies to more general state space as well.  
 The STEEP algorithm bears similarities with the Equi-Energy sampler. One key difference is that STEEP emphasizes the long-range proposal, through which the STEEP provides a new perspective on the advantage of using empirical distribution of tempering.  We note that, at least in principle, the ergodic theorem and the analysis of the spectral gap apply to the Equi-Energy sampler. 
The STEEP also bears similarities with pure tempering methods such as MCMCMC. We point out three key differences.
First, STEEP takes advantage of the empirical distributions, while the ``swap" in tempering methods always use the current state of each chain.  
Second, the interaction of the different chains in STEEP is one-way -- from high temperature chains to the lower ones.  Since the higher temperature chains are not affected by the lower temperature chain, the exploring is more efficient.
Third, tempering relies on local proposals to be more efficient on a \emph{flattened} distribution, while STEEP relies on long-range proposals to be more efficient on a \emph{fattened} distribution.   

The ``Powering-up" is convenient to obtain flattened (or fattened) alternative distributions. However, for certain type of models, e.g., Ising model and its generalization, Potts model, powering-up could run into problems because there might exists a phase transition at a critical temperature \citep{BR04}. When that happens, distributions of above and below the critical temperature may have little similarity; thus it becomes a moot point to use empirical distribution of one temperature to generate long-range proposals for another. To circumvent the ``phase transition" one should discard powering-up scheme. Instead, one may use the ``multi-set sampler" \citep{multiset} to \emph{smooth} a distribution to achieve a similar effect as tempering. A multi-set sampler augments the state space from $\xx$ to $\mathbb{R}^{d \times m} $ so that the current state is a vector $(x_1, \cdots x_m)$. Define $\pi'(x_1, \cdots, x_m) = \frac{1}{m} (\pi(x_1) + \cdots + \pi(x_m))$. This averaging effectively gives a smoother marginal distribution $\pi'_1(\cdot)$ compared to $\pi(\cdot)$. And we can control the smoothness by varying $m$.

In this paper, we focus on the mixing between different modes because within each mode local proposals guarantee rapidly mixing. However, when there exists a mode that is highly correlated among certain dimensions, one needs to fine tune the local proposals. Incorporating certain adaptive sampling scheme such as \citet{yang} for local proposal into STEEP might be desirable.


\bibliographystyle{natbib}
\bibliography{swmcmc,mccr}

\begin{thebibliography}{}

\bibitem[Andrieu {\em et~al.}(2007)Andrieu, Jasra, Doucet, and
  Del~Moral]{andrieu+etal}
Andrieu, C., Jasra, A., Doucet, A., and Del~Moral, P. (2007).
\newblock On nonlinear markov chain monte carlo via self-interacting
  approximations.
\newblock {\em Technical Report\/}.

\bibitem[Andrieu {\em et~al.}(2008)Andrieu, Jasra, Doucet, and
  Del~Moral]{ee-proof}
Andrieu, C., Jasra, A., Doucet, A., and Del~Moral, P. (2008).
\newblock A note on convergence of the equi-energy sampler.
\newblock {\em Stochastic Analysis and Applications\/}, {\bf 26}, 298--312.

\bibitem[Atchad\'{e}(2010)Atchad\'{e}]{atchade10}
Atchad\'{e}, Y. (2010).
\newblock A cautionary tale on the efficiency of some adaptive monte carlo
  schemes.
\newblock {\em Ann. Appl. Probab.}, {\bf 20}, 116Ð154.

\bibitem[Atchad\'{e} and Liu(2006)Atchad\'{e} and Liu]{atchade+liu}
Atchad\'{e}, Y.~F. and Liu, S.~J. (2006).
\newblock Discussion of equi-energy sampler.
\newblock {\em Annals of Statistics\/}, {\bf 34}(4), 1620--1628.

\bibitem[Atchad\'{e} and Rosenthal(2005)Atchad\'{e} and
  Rosenthal]{atchade+rosenthal}
Atchad\'{e}, Y.~F. and Rosenthal, J.~S. (2005).
\newblock On adaptive markov chain monte carlo algorithms.
\newblock {\em Bernoulli\/}, {\bf 11}, 815--828.

\bibitem[Bhatnagar and Randall(2004)Bhatnagar and Randall]{BR04}
Bhatnagar, N. and Randall, D. (2004).
\newblock Torpid mixing of simulated tempering on the {Potts} model.
\newblock In {\em Proceedings of the 15th Annual ACM-SIAM Symposium on Discrete
  Algorithms (New Orleans, LA)\/}, pages 478--487.

\bibitem[Brockwell {\em et~al.}(2010)Brockwell, Del~Moral, and
  Doucet]{brockwell}
Brockwell, A., Del~Moral, P., and Doucet, A. (2010).
\newblock Sequentially interacting markov chain monte carlo.
\newblock {\em Ann. Statist.}, {\bf 38}, 3387--3411.

\bibitem[Craiu {\em et~al.}(2009)Craiu, Rosenthal, and Yang]{yang}
Craiu, R.~V., Rosenthal, J.~S., and Yang, C. (2009).
\newblock Learn from thy neighbor: Parallel-chain adaptive mcmc.
\newblock {\em Preprint\/}.

\bibitem[Davidson and de~Jong(1997)Davidson and de~Jong]{davidson+dejong}
Davidson, J. and de~Jong, R. (1997).
\newblock Strong laws of large numbers for dependent heterogeneous processes: a
  synthesis of recent and new results.
\newblock {\em Econometric Rev.}, {\bf 16}, 251--279.

\bibitem[Del~Moral and Doucet(2010)Del~Moral and Doucet]{delmoral}
Del~Moral, P. and Doucet, A. (2010).
\newblock Interacting markov chain monte carlo methods for solving nonlinear
  measure-valued equations.
\newblock {\em Ann. Appl. Probab.}, {\bf 20}, 593Ð639.

\bibitem[Felsenstein(2004)Felsenstein]{Fbook}
Felsenstein, J. (2004).
\newblock {\em Inferring Phylogenies\/}.
\newblock Sinauer Associates, Sunderland, Mass.

\bibitem[Fort {\em et~al.}(2010)Fort, Moulines, and Priouret]{fort}
Fort, G., Moulines, E., and Priouret, P. (2010).
\newblock Convergence of adaptive mcmc algorithms: ergodicity and law of large
  numbers.
\newblock {\em Technical Report, Paris Tech.}

\bibitem[Geyer(1991)Geyer]{Geyer91}
Geyer, C.~J. (1991).
\newblock Markov chain {M}onte {C}arlo maximum likelihood.
\newblock In E.~M. Keramides, editor, {\em Computing Science and Statistics:
  Proceedings of the 23rd Symposium on the Interface\/}, pages 156--163.
  Interface Foundation, Fairfax Station.

\bibitem[Guan and Krone(2007)Guan and Krone]{guankrone}
Guan, Y. and Krone, S.~M. (2007).
\newblock Small-world mcmc and convergence to multi-modal distributions: From
  slow mixing to fast mixing.
\newblock {\em Annals of Applied Probability\/}, {\bf 17}, 284--304.

\bibitem[Guan {\em et~al.}(2006)Guan, Flei$\ss$ner, Joyce, and Krone]{guan}
Guan, Y., Flei$\ss$ner, R., Joyce, P., and Krone, S.~M. (2006).
\newblock {M}arkov chain {M}onte {C}arlo in small worlds.
\newblock {\em Stat. Comput.}, {\bf 16}, 193--202.

\bibitem[Hastings(1970)Hastings]{Hast70}
Hastings, W.~K. (1970).
\newblock Monte carlo sampling methods using markov chains and their
  applications.
\newblock {\em Biometrika\/}, {\bf 57}, 97--109.

\bibitem[Jarner and Roberts(2007)Jarner and Roberts]{jarner.roberts.07}
Jarner, S.~F. and Roberts, G.~O. (2007).
\newblock Convergence of heavy-tailed monte carlo markov chain algorithms.
\newblock {\em Scandinavian Journal of Statistics\/}, {\bf 34}, 781Ð815.

\bibitem[Kannan {\em et~al.}(1995)Kannan, Lov\'asz, and Simonovits]{KLS}
Kannan, R., Lov\'asz, L., and Simonovits, M. (1995).
\newblock Isoperimetric problems for convex bodies and a localization lemma.
\newblock {\em Discrete Comput. Geom.}, {\bf 13}, 541--559.

\bibitem[Kou {\em et~al.}(2006)Kou, Zhou, and Wong]{ee}
Kou, S.~C., Zhou, Q., and Wong, W.~H. (2006).
\newblock Equi-energy sampler with applications in statistical inference and
  statistical mechanics, with discussion.
\newblock {\em Annals of Statistics\/}, {\bf 34}, 1581--1619.

\bibitem[Lawler and Sokal(1988)Lawler and Sokal]{lawlersokal}
Lawler, G.~F. and Sokal, A.~D. (1988).
\newblock Bounds on the {$L^2$} spectrum for {M}arkov chains and {M}arkov
  processes: a generalization of {C}heeger's inequality.
\newblock {\em Trans. Amer. Math. Soc.}, {\bf 309}, 557--580.

\bibitem[Leman {\em et~al.}(2009)Leman, Chen, and Lavine]{multiset}
Leman, S.~C., Chen, Y., and Lavine, M. (2009).
\newblock The multiset sampler.
\newblock {\em Journal of the American Statistical Association\/}, {\bf
  104}(487), 1029--1041.

\bibitem[Lov\'asz and Simonovits(1993)Lov\'asz and Simonovits]{LS}
Lov\'asz, L. and Simonovits, M. (1993).
\newblock Random walks in a convex body and an improved volume algorithm.
\newblock {\em Random Structures Algorithms\/}, {\bf 4}, 359--412.

\bibitem[Lov\'asz and Vempala(2003)Lov\'asz and Vempala]{LV03}
Lov\'asz, L. and Vempala, S. (2003).
\newblock The geometry of logconcave functions and an {$O^*(n^3)$} sampling
  algorithm.
\newblock Microsoft Research Tech. Rep. MSR-TR-2003--4. Available at :
  \url{http://www-math.mit.edu/~vempala/papers/logcon-ball.ps}.

\bibitem[Madras and Randall(2002)Madras and Randall]{madras}
Madras, N. and Randall, D. (2002).
\newblock Markov chain decomposition for convergence rate analysis.
\newblock {\em Ann. Appl. Probab.}, {\bf 12}, 581--606.

\bibitem[Madras and Zheng(2003)Madras and Zheng]{madras.zheng}
Madras, N. and Zheng, Z. (2003).
\newblock On the swapping algorithm.
\newblock {\em Random Structures and Algorithms\/}, {\bf 1}, 66Ð97.

\bibitem[Marinari and Parisi(1992)Marinari and Parisi]{tempering}
Marinari, E. and Parisi, G. (1992).
\newblock Simulated tempering: a new {M}onte {C}arlo scheme.
\newblock {\em Europhysics Letters\/}, {\bf 19}, 451--458.

\bibitem[Math\'{e} and Novak(2007)Math\'{e} and Novak]{peter}
Math\'{e}, P. and Novak, E. (2007).
\newblock Simple monte carlo and the metropolis algorithm.
\newblock {\em J. Complex.}, {\bf 23}(4-6), 673--696.

\bibitem[Metropolis {\em et~al.}(1953)Metropolis, Rosenbluth, Rosenbluth,
  Teller, and Teller]{metropolis}
Metropolis, N., Rosenbluth, A.~E., Rosenbluth, M.~N., Teller, A.~H., and
  Teller, E. (1953).
\newblock Equation of state calculations by fast computing machines.
\newblock {\em J. Chem. Phys.}, {\bf 21}, 1087--1091.

\bibitem[Meyn and Tweedie(1993)Meyn and Tweedie]{meynbook}
Meyn, S.~P. and Tweedie, R.~L. (1993).
\newblock {\em Markov Chains and Stochastic Stability\/}.
\newblock Springer-Verlag.

\bibitem[Mossel and Vigoda(2005)Mossel and Vigoda]{vigoda}
Mossel, E. and Vigoda, E. (2005).
\newblock {Phylogenetic MCMC Algorithms Are Misleading on Mixtures of Trees}.
\newblock {\em Science\/}, {\bf 309}(5744), 2207--2209.

\bibitem[Neal(2001)Neal]{neal}
Neal, R.~M. (2001).
\newblock Annealed importance sampling.
\newblock {\em Statistics and Computing\/}, {\bf 11}(2), 125--139.

\bibitem[Pe\~na(2005)Pe\~na]{pena}
Pe\~na, J.~M. (2005).
\newblock Exclusion and inclusion intervals for the real eigenvalues of
  positive matrices.
\newblock {\em SIAM Journal on Matrix Analysis and Applications\/}, {\bf
  26(4)}, 908--917.

\bibitem[Predescu {\em et~al.}(2004)Predescu, Predescu, and Ciobanu]{physics}
Predescu, C., Predescu, M., and Ciobanu, C.~V. (2004).
\newblock The incomplete beta function law for parallel tempering sampling of
  classical canonical systems.
\newblock {\em J.CHEM.PHYS.}, {\bf 120}, 4119.

\bibitem[Rambaut and Grass(1997)Rambaut and Grass]{seq-gen}
Rambaut, A. and Grass, N.~C. (1997).
\newblock {Seq-Gen: an application for the Monte Carlo simulation of DNA
  sequence evolution along phylogenetic trees}.
\newblock {\em Comput. Appl. Biosci.}, {\bf 13}(3), 235--238.

\bibitem[Roberts and Rosenthal(1997)Roberts and Rosenthal]{rr}
Roberts, G.~O. and Rosenthal, J.~S. (1997).
\newblock Geometric ergodicity and hybrid {M}arkov chains.
\newblock {\em Electron. Comm. Probab.}, {\bf 2(2)}, 13--25.

\bibitem[Roberts and Rosenthal(2004)Roberts and Rosenthal]{rrsurvey}
Roberts, G.~O. and Rosenthal, J.~S. (2004).
\newblock General state space markov chains and mcmc algorithms.
\newblock {\em Probability Surveys\/}, {\bf 1}, 20--71.

\bibitem[Roberts and Tweedie(2001)Roberts and Tweedie]{roberts}
Roberts, G.~O. and Tweedie, R.~L. (2001).
\newblock Geometric {$L^2$} and {$L^1$} convergence are equivalent for
  reversible {M}arkov chains.
\newblock {\em J. Appl. Probab.}, {\bf 38A}, 37--41.

\bibitem[Ronquist and Huelsenbeck(2003)Ronquist and Huelsenbeck]{mrbayes}
Ronquist, F. and Huelsenbeck, J.~P. (2003).
\newblock Mrbayes 3: Bayesian phylogenetic inference under mixed models.
\newblock {\em Bioinformatics\/}, {\bf 19}, 1572--1574.

\bibitem[Royden(1988)Royden]{roydenbook}
Royden, H.~L. (1988).
\newblock {\em Real Analysis, second edition\/}.
\newblock Collier-Macmillan, London.

\bibitem[Rudolf(2009)Rudolf]{rudolf}
Rudolf, D. (2009).
\newblock Explicit error bounds for lazy reversible markov chain monte carlo.
\newblock {\em Journal of Complexity\/}, {\bf 25}(1), 11 -- 24.

\bibitem[Woodard {\em et~al.}(2008)Woodard, Schmidler, and Huber]{Woodard07}
Woodard, D.~B., Schmidler, S.~C., and Huber, M. (2008).
\newblock Conditions for rapid mixing of parallel and simulated tempering on
  multimodal.
\newblock {\em Annals of Applied Probability\/}.

\bibitem[Yang(1997)Yang]{paml}
Yang, Z. (1997).
\newblock Paml: a program package for phylogenetic analysis by maximum
  likelihood.
\newblock {\em Comput. Appl. Biosci.}, {\bf 15}, 555--556.

\end{thebibliography}

\end{document}